\documentclass[12pt, reqno]{amsart}
\usepackage{amsmath,amsthm,amssymb}

\usepackage[backend=bibtex, style=alphabetic, maxbibnames=99]{biblatex}
\DeclareFieldFormat{postnote}{\mknormrange{#1}}
\DeclareFieldFormat{volcitepages}{\mknormrange{#1}}
\DeclareFieldFormat{multipostnote}{\mknormrange{#1}}

\usepackage[hidelinks]{hyperref}

\newcounter{relctr} 
\everydisplay\expandafter{\the\everydisplay\setcounter{relctr}{0}}

\newcommand\labelrel[2]{
		\begingroup
		\refstepcounter{relctr}
		\stackrel{\textnormal{(\alph{relctr})}}{\mathstrut{#1}}
		\originallabel{#2}
		\endgroup
	}
	\AtBeginDocument{\let\originallabel\label} 
	\DeclareMathOperator*{\argmax}{arg\,max}
	\DeclareMathOperator\Fix{Fix}
	\def\R{\mathbb{R}}
	
	\newtheorem{thm}{Theorem}[section]
	\newtheorem*{thm*}{Theorem}
	\newtheorem{cor}[thm]{Corollary}
	\newtheorem{lm}[thm]{Lemma}
	
	\newtheorem{ft}[thm]{Fact}
	
	\theoremstyle{definition}
	\newtheorem{df}[thm]{Definition}
	\newtheorem{eg}[thm]{Example}
	\theoremstyle{remark}
	\newtheorem{rk}[thm]{Remark}
	
	\usepackage{orcidlink}

	\title{Generalization of Zhou fixed point theorem}
	\date{\today}
	\author{Lu YU\,\orcidlink{0000-0001-6154-4229}}
	\address{Université Paris 1 Panthéon-Sorbonne, UMR 8074, Centre d'Economie de la Sorbonne, Paris, France}\email{yulumaths@gmail.com}

\begin{document}
		
		\begin{abstract}
			We give two generalizations of  the Zhou fixed point theorem.  They weaken the subcompleteness condition of values, and relax the ascending condition of the correspondence. As an application, we derive a generalization of Topkis's theorem on the existence and order structure of the set of Nash equilibria of supermodular games.
		\end{abstract}
		\maketitle

\section{Introduction}
Fixed point theory serves as a bridge across various domains of mathematics and applied sciences, by proving the existence and  uniqueness of solutions in numerous mathematical problems. It has different branches of fixed point theorems that are suitable for distinct types of problems. 	Tarski's fixed point theorem \cite[Theorem 1]{tarski1955lattice} shows that, on a nonempty complete lattice in the sense of Definition \ref{df:lattice} (\ref{it:cplt}), the set of fixed points of an increasing self-map is again a nonempty complete lattice. Under topological semicontinuity, Topkis \cite[Thm.~3.1]{topkis1979equilibrium} applies this to establish the existence of the least and the largest Nash equilibria for supermodular games, and provides two algorithms for seeking a Nash equilibria in a supermodular game in \cite[Sec.~4]{topkis1979equilibrium}. 

Ascendingness of correspondences recalled in Definition \ref{df:inc} (\ref{it:Zasc}) is a generalization of increasingness of single-valued maps.  Zhou  generalizes Tarski's fixed point theorem from increasing maps to ascending correspondences. 

For a correspondence $F:S\to 2^S$ on a set $S$, denote the set of fixed points of $F$ by $\Fix(F)=\{s\in S|s\in F(s)\}$. The notion of subcomplete sublattices is recalled in  Definition \ref{df:lattice} (\ref{it:supcplt}).
\begin{ft}\textup{\cite[Theorem 1]{zhou1994set}}\label{ft:Zhou}
	Let $S$ be a nonempty complete lattice. Let $F:S\to 2^S$ be an ascending correspondence such that for every $s\in S$, the value $F(s)$ is a nonempty subcomplete sublattice of $S$. Then $\Fix(F)$ is a nonempty complete lattice.
\end{ft}
Using Fact \ref{ft:Zhou}, Zhou   \cite[Theorem 2]{zhou1994set}	proves that the set of Nash equilibria of a supermodular game is a nonempty complete lattice.	 Thereby, \cite[Theorem 2]{zhou1994set} generalizes Topkis's theorem. Topkis's and Zhou's theorems  involve topological semicontinuity.   Milgrom and Roberts \cite[Thm.~5]{milgrom1990rationalizability} give a purely order-theoretic existence result of largest and smallest Nash equilibria for supermodular games.  It involves order upper semicontinuity instead, a notion recalled in Definition \ref{df:halforderusc}. Echenique \cite{echenique2005short} gives another constructive proof of Tarski's fixed point theorem and of Fact \ref{ft:Zhou}. Built on this  constructive proof,  \cite[Result 3]{karagozouglu2024submodularity} shows that the Nash equilibria of a supermodular game can be constructed. 

However, the subcompleteness in Fact \ref{ft:Zhou} can be restrictive in applications. As Sabarwal \cite[p.3]{sabarwal2023general} remarks, widely used results in the literature only imply that the set of maximizers of a payoff function in a general model of normal form games is a complete sublattice, which may not be subcomplete. Calciano \cite[Theorem 13]{calciano2010theory} generalizes Fact \ref{ft:Zhou} to correspondences whose values are not necessarily subcomplete. Instead, Calciano requires the intersections with closed intervals to have least elements. A similar hypothesis is used in Sabarwal's generalization \cite[Theorem 4]{sabarwal2023general} of Fact \ref{ft:Zhou}. Furthermore, Calciano and Sabarwal also weaken the ascendingness assumption on the correspondence.    
\subsection*{Main results}
We  relax the subcompleteness restriction in Fact \ref{ft:Zhou}. Theorem \ref{thm:myZhou} is a generalization of Zhou's fixed point theorem to correspondences whose values may not be sublattices, and the subcompleteness hypothesis is relaxed to chain-subcompleteness in the sense of Definition \ref{df:lattice} (\ref{it:chainsubcplt}). These generalizations allows more applications in game theory. The concept of V-ascending correspondence are defined in Definition \ref{df:inc} (\ref{it:Vasc}).
\begin{thm}\label{thm:myZhou}
	Let $S$ be a nonempty complete lattice. Let $F:S\to 2^S$ be a  V-ascending correspondence. Suppose that for every $x\in S$, the subset $F(x)$ is  nonempty and chain-subcomplete in $S$.  Then $\Fix(F)$ is a nonempty complete lattice.
\end{thm}
Theorem \ref{thm:cpltval} is another generalization of Fact \ref{ft:Zhou}. It  weakens the subcompleteness of the values to   completeness.  We emphasis that subcompleteness is a relative notion, as it depends on the ambient larger lattice, while completeness is an intrinsic property the the lattice.
\begin{thm}\label{thm:cpltval}
	Let $S$ be a nonempty complete lattice. Let $F:S\to 2^S$ be a V-ascending correspondence. Suppose that for every $x\in S$, the value $F(x)$ is a nonempty complete lattice. Then $\Fix(F)$ is a nonempty complete lattice.
\end{thm}Theorem \ref{thm:cpltval} allows wider applications than Fact \ref{ft:Zhou}. For example,  we prove
Theorem \ref{thm:introgame}, which extends Milgrom and Robert's theorem from product-form games to possibly non-product-form games.
\begin{thm}\label{thm:introgame}
	Given a supermodular game in the form of \eqref{eq:supgame},	assume that for every $i\in N$ and  every $x\in S$, 
	the function $f_i(\cdot,x_{-i}):S_i(x_{-i})\to \R$ is order upper semicontinuous. Then the set of Nash equilibria is a nonempty complete lattice.
\end{thm}

\section{Preliminaries} We recall necessary definitions, provide examples of Theorems \ref{thm:myZhou} and \ref{thm:cpltval}, and demonstrate several preparatory lemmas. A set with a partial order is called a \emph{poset}.

\begin{df}\label{df:lattice}
	Let $S$ be a poset. \begin{enumerate}
		\item\label{it:supcplt} Let $T$ be a subset of $S$.  If for every nonempty subset $A\subset T$, both $\sup_S(A)$ and $\inf_S(A)$ exist and belong to $T$, then $T$ is called a \emph{subcomplete sublattice} of $S$.
		\item\label{it:cplt} The poset $S$ is called a \emph{complete lattice} if it is a subcomplete   sublattice of itself.
		\item The poset $S$ is called \emph{join-complete} if for every nonempty subset $A\subset S$, the supremum $\sup_SA$ exists.
		\item\label{it:chainsubcplt} (\cite[p.1420]{heikkila2006fixed}) Let $T$ be a subset of $S$.  If for every nonempty chain $C\subset T$,  the element $\sup_S(C)$(resp. $\inf_S(C)$) exists and belongs to $T$, then $T$ is called \emph{chain-subcomplete upwards}(resp. \emph{downwards}) in $S$. If $T$ is chain-subcomplete upwards and downwards, then $T$ is called \emph{chain-subcomplete} in $S$.
	\end{enumerate}
\end{df}
A subcomplete sublattice is chain-subcomplete, but the converse is not true. 

\begin{df}\label{df:inc} Let $X$ be a poset, and let $Y$ be a lattice. Let $F:X\to 2^Y$ be a correspondence.
	\begin{enumerate}
		\item \label{it:Zasc}(\cite[p.296]{zhou1994set}) If for any $x\le x'$ in $X$, every $y\in F(x)$ and every $y'\in F(x')$, one has $y\wedge y'\in F(x)$ and $y\vee y'\in F(x')$, then $F$ is called \emph{ascending}.
		\item\label{it:Vasc} (\cite[Ch.~4, Sec.~3]{veinott1992lattice}) If for any  $x<x'$ in $X$,\footnote{Note the strict inequality here.} every $y\in F(x)$ and $y'\in F(x')$,  we have 	$y\wedge y'\in F(x)$(resp. $y\vee y'\in F(x')$), then $F$ is called \emph{lower}(resp. \emph{upper}) \emph{V-ascending}. If $F$ is both upper and lower  V-ascending, then $F$ is called \emph{V-ascending}.
	\end{enumerate}
\end{df}
An ascending correspondence is V-ascending, but not vice versa. 

We give some examples to show the relation between Fact \ref{ft:Zhou}, Theorems \ref{thm:myZhou} and \ref{thm:cpltval}.	\begin{eg}
	Let $S=\{0,1,a,b\}$ be a lattice, where $0=\min S$, $1=\max S$ and $a,b$ are incomparable. Define a correspondence $F:S\to 2^S$ by $F(0)=\{0\}$, $F(1)=\{1\}$ and $F(a)=F(b)=\{a,b\}$. Then $F$ is V-ascending but not ascending. The value $F(a)$ is not a lattice. Thus, one cannot apply Fact \ref{ft:Zhou} nor Theorem \ref{thm:cpltval} to $F$. Since $F(a)$ has neither maximum nor minimum, one cannot apply \cite[Thm.~4]{sabarwal2023general}, neither.  Theorem \ref{thm:myZhou} is applicable in this case, so it is  a proper generalization of Fact \ref{ft:Zhou}.
\end{eg}	\begin{eg}
	Let $S=\{0,a,b,1,2\}$, where $0\le a,b\le 1<2$ and $a,b$ are incomparable. Then $S$ is a complete lattice. Define a correspondence $F:S\to 2^S$ by $F(0)=\{0\}$, $F(a)=F(b)=\{0,a,b,2\}$, $F(1)=\{1,2\}$ and $F(2)=\{2\}$. Then $F$ is V-ascending but not ascending. The value $F(a)$ is not a sublattice of $S$. Thus, Fact \ref{ft:Zhou} does not apply to $F$. Still, Theorem \ref{thm:cpltval} shows that $\Fix(F)$ is a nonempty complete lattice. Therefore, Theorem \ref{thm:cpltval} is also a proper generalization of Fact \ref{ft:Zhou}.
\end{eg}
\begin{eg}
	Let $S=[0,3]$. Define a correspondence $F:S\to 2^S$ by  \[F(x)=\begin{cases}
		\{0\} &\text{if }x<1,\\
		[0,1)\cup\{3/2\} &\text{if } x=1,\\
		\{3/2\} &\text{if } 1<x<2,\\
		\{3/2\}\cup(2,3]&\text{if }x=2,\\
		\{3\}&\text{if }x>2.
	\end{cases}\] Then $F$ is ascending, but $F(1)$ is not chain-subcomplete in $S$. Thus, Theorem \ref{thm:myZhou} is not applicable to $F$. Moreover, the restriction $F':[2,3]\to 2^{[2,3]}, \, x\mapsto F(x)\cap [2,3]$ is not an isotone infimum model in the sense of \cite[p.9]{sabarwal2023general}, so $F$ is not isotone infimum on upper intervals in the sense of \cite[p.22]{sabarwal2023general}. Similarly, considering another restriction $F'':[0,1]\to 2^{[0,1]},\, x\mapsto F(x)\cap [0,1]$, one obtains that $F$ is not isotone infimum on lower intervals. Consequently, $F$ is not in the two classes of lattice models in \cite[Thm.~4]{sabarwal2023general}. Still, $F$ satisfies the hypothesis of Theorem \ref{thm:cpltval}. Therefore, Theorem \ref{thm:cpltval} is not covered by Theorem \ref{thm:myZhou} nor Sabarwal's theorem.
\end{eg}
Several lemmas used in the proof of the main results are as follows. 
\begin{lm}\label{lm:jointmin}
	A join-complete poset  with a least element  is a complete lattice.
\end{lm}
\begin{proof}
	Let $X$ be join-complete poset  with a least element.	For every nonempty subset $Y$ of $X$, the subset $Z:=\{x\in X:x\le y,\forall y\in Y\}$ contains $\min X$. As $Z$ is nonempty and $X$ is join-complete, the element $u:=\sup_XZ$ exists. For   every $z\in Z$ and every $y\in Y$, we have $z\le y$, so $u\le y$. Therefore, $u\in Z$ and hence $u=\max Z=\inf_XY$. Thus $X$ is complete.
\end{proof}
Lemma \ref{lm:Veinott} is from Veinott unpublished note \cite{veinott1992lattice}. For the convenience of the reader, we reproduce the proof.
\begin{lm}[Veinott]\label{lm:Veinott}
	Let $(S,\le)$ be a poset. Let $A,B$ be two subsets of $S$. Assume that $B$ is chain-subcomplete downwards(resp. upwards). Fix $x\in A\cap B$. If for every $a\in A\setminus\{x\}$ and every $b\in B$, the element $a\wedge b$(resp. $a\vee b$) exists and is in $B$, then $\inf_S(A)$(resp. $\sup_S(A)$) exists and is in $B$.
\end{lm}
\begin{proof}
	By symmetry, it suffices to prove the statement without parentheses.	The statement holds when $A=\{x\}$. Assume that $A\setminus\{x\}$ is nonempty. By the well ordering principle, there is a well order $\preceq$ on the nonempty set $A\setminus\{x\}$. Let $T$ be a copy of $A$ together with an exotic element $\infty$. Extend the order $\preceq$ from $A\setminus\{x\}$ to $T$ with $\min T=x$ and $\max T=\infty$. Then $(T,\preceq)$ is still a well ordered set. Define a map $f:T\to A$ by \[f(t)=\begin{cases}
		t &\text{if } t\in A,\\
		x &\text{if } t=\infty.
	\end{cases}\]
	For every $t\in T$, let $A_t=\{f(t'):t'\le t\}$.
	We  show that $\inf_S(A_t)$ exists and belongs to $B$ for every $t\in T$ by transfinite induction.
	
	For the base case $t=x$, one has $A_x=\{x\}$, then $\inf_S(A_x)=x$ is in $B$.
	
	For $t_0(\neq x)\in T$,  assume the statement is proved for for all $t\prec t_0$. We prove the statement for $t_0$. For every $t\preceq t'\prec t_0$, one has $A_{t}\subset A_{t'}$. By the inductive hypotheses, $\inf_S(A_t)$ and $\inf_S(A_{t'})$ exist and are in $B$. Moreover, $\inf_S(A_t)\ge \inf_S(A_{t'})\ge x$. Therefore, $\{\inf_S(A_t)\}_{t\prec t_0}$ is a nonempty chain in $B$. As $B$ is chain-subcomplete downwards in $S$, their infimum $m:=\inf_S\{\inf_S(A_t)\}_{t\prec t_0}$ exists and is in $B$. Moreover, $m\ge x$. If $f(t_0)=x$, then $f(t_0)\wedge m=f(t_0)=x\in B$. If $f(t_0)\neq x$, then $f(t_0)\wedge m\in B$ by assumption. In both cases, $f(t_0)\wedge m$ exists an is in $B$.
	
	Since $A_{t_0}=\{f(t_0)\}\cup\cup_{t\prec t_0}A_t$, we get $\inf_S(A_{t_0})=f(t_0)\wedge m\in B$, i.e., the statement holds for $t_0$. The induction is completed. The statement for $t=\infty$ gives that $\inf_S(A)$ exists and is in $B$. 
\end{proof}
\begin{rk}
	The chain-subcomplete upwards condition in Lemma \ref{lm:Veinott} cannot be weakened to  chain-complete upwards condition. For example, take $S=\R$, $A=[0,1)$ and $B=[0,1)\cup\{2\}$ and $x=0$. Then $B$ is chain-complete upwards, but $\sup_S(A)=1\notin B$.
\end{rk}	\begin{lm}\label{lm:Zhou}
	Let $S$ be a nonempty complete lattice. Let $F:S\to2^S$ be a lower V-ascending correspondence, such that for every $x\in S$, the value $F(x)$ is  nonempty and chain-subcomplete downwards in $S$.  Then  $F$  has a least fixed point.
\end{lm}
\begin{proof}Since $S$ is complete, the element $m:=\max S$ exists. 
	Let $C=\{c\in S:\exists x\in F(c),x\le c\}$. Because $F(m)$ is nonempty, there is  $x\in F(m)$. Then $x\le m$, so $m\in C$. In particular, $C$ is nonempty. By the completeness of $S$ again, $a=\inf_S(C)$ exists. 
	
	We show that $a\in C$. Otherwise, $a\notin C$. Because $F(a)$ is nonempty, we may fix $\alpha\in F(a)$. Let $A=\{\alpha\}\cup \cup_{c\in C}F(c)$. 
	For every $c\in C$, one has $a<c$. For every $d\in F(c)$ and every $\alpha'\in F(a)$, because $F$ is lower V-ascending, we have $d\wedge \alpha'\in F(a)$. Because $F(a)$ is chain-subcomplete downwards in $S$, by Lemma \ref{lm:Veinott}, one has $\inf_S(A)\in F(a)$.  Since $c\in C$, there is $x_c\in F(c)$ with $c\ge x_c$. One has $x_c\ge \inf_S(A)$. Then $\inf_S(A)\le c$ for all $c\in C$. Therefore, $\inf_S(A)\le a$ and $a\in C$, a contradiction.
	
	We prove  $a\in \Fix(F)$. Assume the contrary $a\notin  \Fix(F)$. From last paragraph, there is $x_a\in F(a)$ with $x_a\le a$. Then $x_a<a$. Since $F(x_a)\neq\emptyset$, there is  $y\in F(x_a)$. As $F$ is lower V-ascending, one has $y\wedge x_a\in F(x_a)$. Therefore, $x_a\in C$. Thus, $a\le x_a$, a contradiction. The claim is proved.  In particular, $\Fix(F)$ is nonempty.
	
	For every $e\in \Fix(F)$, one has $e\in C$, so $e\ge a$, which proves $a=\min \Fix(F)$. 
\end{proof}	\begin{lm}\label{lm:leastfix}
	Let $S$ be a nonempty complete lattice. Let $F:S\to 2^S$ be a lower V-ascending correspondence. Suppose that for every $x\in S$, the value $F(x)$ has a least element. Then $F$ has a least fixed point. 
\end{lm}
\begin{proof}
	Let $A:=\{x\in S:\min F(x)\le x\}$. Since $S$ is  nonempty complete, $\max S$ exists and belongs to $A$. In particular, $A$ is nonempty. Thus, the element $x_*:=\inf_SA$ exists. 
	
	Set $m:=\min F(x_*)$. We prove that $x_*\in A$. Otherwise, $x_*\notin A$. Then for every $a\in A$, one has $a>x_*$ and $\min F(a)\le a$. Because $F$ is lower V-ascending, $\min F(a)\wedge m\in F(x_*)$. Since $\min F(a)\wedge m\le m$, one has \[m=\min F(a)\wedge m\le \min F(a)\le a.\] Therefore, $m\le x_*$, a contradiction. 
	
	We prove that $x_*\in\Fix(F)$. Otherwise, $m<x_*$. Because $F$ is lower V-ascending, $\min F(m)\wedge m\in F(m)$. Since $\min F(m)\wedge m\le \min F(m)$, one has $\min F(m)=\min F(m)\wedge m\le m$. Therefore, $m\in A$ and hence $m\ge x_*$, a contradiction.
	
	For every  $x\in\Fix(F)$, one has $x\in F(x)$, so $\min F(x)\le x$. Thus, $x\in A$ and hence $x\ge x_*$. \end{proof}
\section{Proof of Theorem \ref{thm:myZhou}}
As $F$ is lower V-ascending and the values are chain-subcomplete downwards, by Lemma \ref{lm:Zhou}, the subset $\Fix(F)$ is nonempty and has a least element. We show that $\Fix(F)$ is join-complete, i.e., for every nonempty subset $U\subset \Fix(F)$,  the element $\sup_{\Fix(F)}(U)$ exists.

By the completeness of $S$, the supremum $b:=\sup_SU$ exists. If $b\in \Fix(F)$, then $\sup_{\Fix(F)}(U)$ is $b$. Now assume that $b\notin \Fix(F)$. Since $F(b)$ is nonempty, one may take $\beta\in F(b)$. Let $A=\{\beta\}\cup U$. For every $u\in U$, we have $u\le b$. As $b\notin \Fix(F)$, we have a strict inequality $u<b$. For every $\beta'\in F(b)$, since $F$ is upper V-ascending, one has $u\vee \beta'\in F(b)$. Because $F(b)$ is chain-subcomplete upwards in $S$, by Lemma \ref{lm:Veinott}, $\sup_S(A)\in F(b)$. Since $A\supset U$, we have  $\sup_S(A)\ge b$.

Let $S':=[b,+\infty)_S$ be a closed interval of $S$. Then $S'$ is a nonempty subcomplete sublattice of $S$.
Define a correspondence \[F':S'\to 2^{S'},\quad F'(s)=F(s)\cap S'.\] 

We show that for every $s\in S'$, the value $F'(s)$ is nonempty. In fact, if $s=b$, then $\sup_S(A)\in F'(s)$. If $s>b$, as $F(s)\neq\emptyset$, there is $t\in F(s)$. Since $F$ is upper V-ascending, one has $t\vee \sup_S(A)\in F(s)$ and $b\le \sup_S(A)\le t\vee \sup_S(A)$. Then $t\vee \sup_S(A)\in F'(s)$. 

For every $s\in S'$, since $F(s)$ is chain-subcomplete downwards in $S$, the value $F'(s)$ is chain-subcomplete downwards in $S'$. 

We prove that the correspondence $F'$ is lower V-ascending as follows.  For any $x< y$ in $S'$, every $s\in F'(x)$ and every $t\in F'(y)$, because $F$ is lower V-ascending, one has $s\wedge t\in F(x)$ and $s\wedge t\ge b$. Thus, one has $s\wedge t\in F'(x)$. 

Therefore, by Lemma \ref{lm:Zhou}, the set $\Fix(F')$  is nonempty and has a minimum element $m$.

Since $m\in F'(m)\subset F(m)$, one gets $m\in \Fix(F)$. For every $u\in U$, one has $u\le b\le m$. For every $e\in \Fix(F)$ with $e\ge u$ for all $u\in U$, we have $e\in S'$. Then $e\in F(e)\cap S'=F'(e)$. One has $e\in \Fix(F')$ and $e\ge m$. Therefore, $\sup_{\Fix(F)}(U)=m$. 

The completeness of $\Fix(F)$ follows from Lemma \ref{lm:jointmin}.

\section{Proof of Theorem \ref{thm:cpltval}}
For every $x\in S$, the value $F(x)$ is complete, so $\min F(x)$ exists. By Lemma \ref{lm:leastfix}, $F$ has a least fixed point. In particular, the set $\Fix(F)$ is nonempty.

For every nonempty subset $T$ of $\Fix(F)$, we show that $\sup_{\Fix(F)}T$ exits.  Since $S$ is complete, the element $b:=\sup_ST$ exists. If $b\in \Fix(F)$, then $\sup_{\Fix(F)}T=b$. Thus, one may assume  $b\notin\Fix(F)$. Then $S':=[b,+\infty)_S$ is a closed interval in $S$, hence a nonempty complete lattice. Define a correspondence $F':S'\to 2^{S'}$ by $F'(x)=F(x)\cap S'$.

For every $x\in S'$ and every $t\in T$, one has $t\le b$. Since $b\notin \Fix(F)$, one has $t<b$. Because $t\in \Fix(F)$, one has $t\in F(t)$. As $F$ is upper V-ascending, $t\vee \min F(x)\in F(x)$. By the completeness of $F(x)$, the element $v(x):=\sup_{F(x)}\{t\vee \min F(x):t\in T\}$ exists. For every $t\in T$, one has $t\le t\vee \min F(x)\le v(x)$. Thus, $b\le v(x)$ and $v(x)\in F'(x)$. In particular, $F'(x)$ is nonempty.

We prove that $v(x)=\min F'(x)$. For every $y\in F'(x)$, $y\ge \min F(x)$. For every $t\in T$, $y\ge b\ge t$, so $y\ge t\vee \min F(x)$. Since $y\in F(x)$, one has $y\ge v(x)$. 

We check that $F'$ is lower V-ascending. For any $x<x'$ in $S'$, every $y\in F'(x)$ and every $y'\in F'(x')$, since $F$ is lower V-ascending, one has $y\wedge y'\in F(x)$. Since $b\le y$ and $b\le y'$, one has $b\le y\wedge y'$ and $y\wedge y'\in S'$. Thus, $y\wedge y'\in F'(x)$. 

By Lemma \ref{lm:leastfix}, the correspondence $F'$ has a least fixed point $b'$. Since $\Fix(F')=\Fix(F)\cap S'$, one gets $b'=\sup_{\Fix(F)}T$. Therefore, $\Fix(F)$ is join-complete. By Lemma \ref{lm:jointmin}, the poset $\Fix(F)$ is a nonempty complete lattice.
\section{Nash equilibria of supermodular games}
We review the classical definition of supermodular games and prove Theorem \ref{thm:introgame}. It is a generalization of Topkis's theorem, as Corollary \ref{cor:Topkis421} shows. 
\begin{df}\label{df:supermodulargameinwkZhou}
	A supermodular game \begin{equation}\label{eq:supgame}(N,\{S_i\},S,\{f_i\})\end{equation} consists of the following data:
	\begin{itemize}
		\item a nonempty finite set of players $N$,
		\item for each $i\in N$, a nonempty  lattice $S_i$ of the strategies of player $i$,
		\item  a nonempty  sublattice\footnote{A product of posets is a poset under the product order.} $S\subset \prod_{i\in N}S_i$ comprised of feasible joint strategies such  that every projection $S\to S_i$ is surjective,
		\item for  every $i\in N$, a payoff function $f_i: S\to \R$ 	\end{itemize}  such that 	for every $i\in N$, 
	\begin{enumerate}\item for every $x_{-i}\in S_{-i}:=\prod_{j\neq i}S_j$, the function $f_i(\cdot,x_{-i}):S_i(x_{-i})\to \R$ is  supermodular, where $S_i(x_{-i})$ is the sublattice $\{x_i\in S_i:(x_i,x_{-i})\in S\}$ of $S_i$;  
		\item  the function $f_i$ has increasing difference relative to the subset $S$ of $S_i\times S_{-i}$ in the sense of \cite[p.42]{topkis1998supermodularity}..
	\end{enumerate}
\end{df}
Fix  a supermodular game \eqref{eq:supgame}.  For $x\in S$, put $S(x)=(\prod_{i\in N}S_i(x_{-i}))\cap S$.
\begin{lm}\label{lm:subcplt}Suppose that  $S\subset \prod_{i\in N}S_i$  is subcomplete. The following results hold.
	\begin{enumerate}\item\label{it:Sicplt} For every $i\in N$, the poset $S_i$ is a complete lattice.
		\item\label{it:Six-i} For every $i\in N$ and every $x_{-i}\in S_{-i}$, the set $S_i(x_{-i})$ is a subcomplete sublattice of $S_i$.
		\item\label{it:Sx} For every $x\in S$, the subset $S(x)$ is a subcomplete sublattice of $S$.\end{enumerate}\end{lm}
\begin{proof}\hfill
	\begin{enumerate}\item Let $T_i$ be a nonempty subset of $S_i$, and $T$ be the preimage of $T_i$ under the projection $S\to S_i$. Since $S\to S_i$ is surjective, the set $T$ is nonempty. Since $S\subset\prod_{j\in N}S_j$ is subcomplete, the element $x=\sup_{\prod_{j\in N}S_j}T\in S$. 
		
		We check that $x_i=\sup_{S_i}T_i$. For every $t_i\in T_i$, there is $t\in T$ whose $i^{\text{th}}$ coordinate is $t_i$. Then $x\ge t$, so $x_i\ge t_i$. Consider an arbitrary $y_i\in S_i$ with $y_i\ge \tau_i$ for all $\tau_i\in T_i$. For every $t\in T$, we have $(y_i,x_{-i})\ge t$. Thus, we have $(y_i,x_{-i})\ge x$ and $y_i\ge x_i$. 
		
		Similarly, the element $\inf_{S_i}T_i$ exists. Therefore, $S_i$ is complete.
		\item  For every nonempty subset $A\subset S_i(x_{-i})$ and every $a\in A$, one has $(a,x_{-i})\in S$. Since $S$ is subcomplete in $\prod_{i\in N}S_i$, by Part (\ref{it:Sicplt}), one has $(\sup_{S_i}A,x_{-i})=\sup_{\prod_{i\in N}S_i}\{(a,x_{-i})\}_{a\in A}\in S$, so $\sup_{S_i}A\in S_i(x_{-i})$. Similarly, one has $\inf_{S_i}A\in S_i(x_{-i})$. 
		
		\item For every nonempty subset $B\subset S(x)$ and  every $i\in N$, let $B_i$ be the image of $B$ under the projection $S\to S_i$. Then $B\subset S_i(x_{-i})$. By Part (\ref{it:Six-i}), one has $\sup_{S_i}B_i\in S_i(x_{-i})$. The sublattice $S$ is subcomplete in $\prod_{i\in N}S_i$, so $\sup_SB=(\sup_{S_i}B_i)_{i\in N}\in (\prod_{i\in N}S_i(x_{-i}))\cap S=S(x)$.\end{enumerate}
\end{proof}
\begin{df}\label{df:wkzhouequi}
	A strategy $s\in S$ is called a Nash equilibrium of the fixed supermodular game if for every $i\in N$ and every $s_i'\in S_i(s_{-i})$, one has $f_i(s'_i,s_{-i})\le f_i(s)$.
\end{df}
\begin{df}[{\cite[p.1260]{milgrom1990rationalizability}}, {\cite[Definition 1]{prokopovych2017strategic}}]\label{df:halforderusc}
	Given a \emph{complete} lattice $X$, a function $f:X\to \R$ is \emph{upward}(resp. \emph{downward}) \emph{upper semicontinuous} if for every nonempty chain $C$ in $X$,  one has \[\limsup_{x\in C,x\to\sup_X(C)}f(x)\le f(\sup_X(C))\]
	\[\text{(resp. }\limsup_{x\in C,x\to\inf_X(C)}f(x)\le f(\inf_X(C))).\]  If $f$ is both upward and downward upper semicontinuous, then it is called \emph{order upper semicontinuous}.
\end{df}
Theorem \ref{thm:introgame} follows from Theorem \ref{thm:order4.2.1} and Lemma \ref{lm:compareconditions}.
\begin{thm}\label{thm:order4.2.1}
	Suppose that  $S\subset \prod_{i\in N}S_i$  is subcomplete.	Assume that for every $i\in N$ and  every $x\in S$, 
	\begin{enumerate}
		\item\label{it:upusc}  the function $f_i(\cdot,x_{-i}):S_i(x_{-i})\to \R$ is upward  upper semicontinuous,
		\item\label{it:lowerbound} and for  every nonempty chain $C\subset S(x)$, there is $b\in S(x)$ such that $b\le c$  for every $c\in C$ and $f_i(b_i,x_{-i})\ge \limsup_{c\in C,c\to \inf_SC}f_i(c_i,x_{-i})$.
	\end{enumerate}
	Then, the set of Nash equilibria is a nonempty complete lattice.
\end{thm}
\begin{proof}
	Because $N$ is finite, 	 one may define a function \[g(\cdot,x):S(x)\to \R,\quad y\mapsto \sum_{i\in N}f_i(y_i,x_{-i})\]for every $x\in  S$. For every $a\in \R$, we prove that the subset $L^a:=\{y\in S(x):g(y,x)\ge a\}$ is chain-subcomplete upwards in $S(x)$  and chain-bounded below. For every nonempty chain $C\subset L^a$, by Lemma \ref{lm:subcplt} (\ref{it:Sx}), one has $\sup_SC\in S(x)$. Then
	\begin{align*}
		g(\sup_S(C),x)=&\sum_{i\in N}f_i((\sup_SC)_i,x_{-i})=\sum_{i\in N}f_i(\sup_{S_i}(C_i),x_{-i})\\
		\labelrel\ge{myeq:upusc} &\sum_{i\in N}\limsup_{y_i\in C_i,y_i\to \sup_{S_i}C_i}f_i(y_i,x_{-i})\\
		=&\sum_{i\in N}\limsup_{y\in C,y\to \sup_SC}f_i(y_i,x_{-i})\\
		\ge &\limsup_{y\in C,y\to \sup_SC} \sum_{i\in N}f_i(y_i,x_{-i})\\
		=&\limsup_{y\in C,y\to \sup_SC}g(y,x)\ge a,\end{align*} where \eqref{myeq:upusc} uses Assumption (\ref{it:upusc}). One gets $\sup_SC\in L^a$. Therefore, $L^a$ is chain-subcomplete upwards in $S(x)$. Take $b$ as in Assumption (\ref{it:lowerbound}). Then \begin{align*}g(b,a)=&\sum_{i\in N}f_i(b_i,x_{-i})\\
		\ge &\sum_{i\in N}\limsup_{c\in C,c\to \inf_SC}f_i(c_i,x_{-i})\\
		\ge &\limsup_{c\in C,c\to \inf_SC}\sum_{i\in N}f_i(c_i,x_{-i})\\
		=&\limsup_{c\in C,c\to \inf_SC}g(c,x)\ge a.\end{align*}  Thus, the chain $C$ has a lower bound $b\in L^a$. By Zorn's lemma, $L^a$ has a minimal element whenever it is nonempty. 
	
	We check that $g(\cdot,x):S(x)\to \R$ is supermodular. For every $y,y'\in S(x)$, $g(y\vee y',x)+g(y\wedge y',x)=\sum_{i\in N}f_i(y_i\vee y'_i,x_{-i})+f_i(y_i\wedge y'_i,x_{-i})\ge \sum_{i\in N}f_i(y_i,x_{-i})+f_i(y'_i,x_{-i})=g(y,x)+g(y',x)$.

	Then by \cite[Corollary 3.4 (3)]{yu2023topkis3normal}, the set $Y(x):=\argmax_{y\in S(x)}g(y,x)$ is nonempty. From last paragraph, the subset $Y(x)=L^{\max_{S(x)}g(\cdot,x) }$ is a chain-subcomplete upwards in $S(x)$  and chain-bounded below. 
	
	We check that $g$ has increasing difference. Consider  $y<y'\in S$, $x<x'\in S$, with $\{y,y'\}\times \{x,x'\}\subset \{(v,u)\in S\times S:v\in S(u)\}$. Since $\{y_i,y'_i\}\times \{x_{-i},x'_{-i}\}\subset S$, we have $f_i(y'_i,x_{-i})+f_i(y_i,x'_{-i})\le f_i(y_i,x_{-i})+f_i(y'_i,x'_{-i})$. Then $g(y',x)+g(y,x')=\sum_{i\in N}f_i(y'_i,x_{-i})+f_i(y_i,x'_{-i})\le \sum_{i\in N}f_i(y_i,x_{-i})+f_i(y'_i,x'_{-i})=g(y,x)+g(y',x')$. By \cite[Lemma 3.10 (2)]{yu2023topkis1} and \cite[Theorem 2.8.1]{topkis1998supermodularity}, the correspondence $Y:S\to 2^S$ is increasing. For every $x\in S$, the value $Y(x)$ is a lattice. By \cite[Corollary 2.7]{yu2023topkis2}, the lattice $Y(x)$ is complete.

	By Theorem \ref{thm:cpltval}, the poset $\Fix(Y)$ is a nonempty complete lattice. We conclude by \cite[Lemma 3.7]{yu2023topkis1}.
\end{proof}
\begin{rk}By Lemma \ref{lm:subcplt} (\ref{it:Six-i}), the upward upper semicontinuity in Assumption (\ref{it:upusc}) of Theorem \ref{thm:order4.2.1} makes sense.\end{rk}

\begin{lm}\label{lm:compareconditions}
	Suppose that  $S\subset \prod_{i\in N}S_i$  is subcomplete. If for every $i\in N$ and  every $x\in S$, the function $f_i(\cdot,x_{-i}):S_i(x_{-i})\to \R$ is downward upper semicontinuous, then Assumption (\ref{it:lowerbound}) of Theorem \ref{thm:order4.2.1} is satisfied.
\end{lm}
\begin{proof}
	Let $b=\inf_SC$.  For every $i\in N$, let $C_i$ be the image of $C$ under the projection $S\to S_i$. Then   $C_i$ is a nonempty chain in $S_i(x_{-i})$ and $b_i=\inf_{S_i}C_i$. Since $f_i(\cdot,x_{-i}):S_i(x_{-i})\to \R$ is downward upper semicontinuous, one has\begin{align*}f_i(b_i,x_{-i})\ge&\limsup_{y_i\in C_i,y_i\to \inf_{S_i(x_{-i})}C_i}f_i(y_i,x_{-i})\\
		\labelrel={myeq:Six-i} &\limsup_{y_i\in C_i,y\to \inf_{S_i}C_i}f_i(y,x_{-i})\\
		=&\limsup_{c\in C,c\to b}f_i(c_i,x_{-i}),\end{align*} where \eqref{myeq:Six-i} uses Lemma \ref{lm:subcplt} (\ref{it:Six-i}). Since $C\subset S(x)$, by Lemma \ref{lm:subcplt} (\ref{it:Sx}), one has $b\in S(x)$.
\end{proof} 
Lemma \ref{lm:cpt>chainsub} is used in the proof of Corollary \ref{cor:Topkis421}.
\begin{lm}\label{lm:cpt>chainsub}
	Let $n\ge0$ be an integer. Then, every compact subset of $\R^n$ is chain-subcomplete.
\end{lm}
\begin{proof}
	Let $K$ be a compact subset of $\R^n$. Let $C$ be a nonempty chain in $K$. Because $K$ is bounded,   the number $y_i=\sup_{\R}\{x_i\}_{x\in C}$ exists for every $1\le i\le n$. Then $y=\sup_{\R^n}C$. 
	
	We prove that $y$ is a limit point of $C$. For every $\epsilon>0$ and every $1\le i\le n$, there is $x^i\in C$ with $y_i-x^i_i<\epsilon/n$. Since $C$ is a chain, there is $1\le j\le n$ with $x^j=\max_{k=1}^nx^k$. Then for every $1\le k\le n$, one has $0\le y_k-x^j_k\le y_k-x_k^k<\epsilon/n$. Thus, the $1$-norm $\Vert y-x^j\Vert_1:=\sum_{k=1}^n|y_k-x^j_k|<\epsilon$.
	
	Since $K$ is closed in $\R^n$, one has $y\in K$. Similarly, the element $\inf_{\R^n}C$ exists in $K$.
\end{proof}Lemma \ref{lm:Euclidusc} (\ref{it:usc>orderusc}) shows that the topological upper semicontinuity is stronger than the order counterpart. An analog for upper semicontinuity in the interval topology is given in \cite[Proposition 3.7]{yu2023topkis2}. As the Euclidean topology on $X$ is finer than the interval topology, the upper semicontinuity in Lemma \ref{lm:Euclidusc} is weaker than that in \cite[Proposition 3.7]{yu2023topkis2}.
\begin{lm}\label{lm:Euclidusc}Let $n\ge0$ be an integer, and let $X$ be a compact sublattice of $\R^n$. Let $f:X\to \R$ be an upper semicontinuous function. Then
	\begin{enumerate}
		\item\label{it:cpt>cplt} the sublattice $X$ is subcomplete in $\R^n$, and
		\item\label{it:usc>orderusc}  every upper  semicontinuous function $f:X\to \R$  is order upper semicontinuous.
	\end{enumerate} 
\end{lm}
\begin{proof}\hfill\begin{enumerate}\item By Lemma \ref{lm:cpt>chainsub} and Veinott's theorem \cite[Theorem 2.6]{yu2023topkis2}, the lattice $X$ is complete.
		\item For every $a\in \R$, the subset $L^a:=\{x\in X:f(x)\ge a\}$ is closed in $X$. Since $X$ is compact, so is $L^a$. By Lemma \ref{lm:cpt>chainsub}, $L^a$ is chain-subcomplete in $\R^n$. By \cite[Proposition 3.4]{yu2023topkis2}, $f$ is order upper semicontinuous. \end{enumerate}
\end{proof}
As an application of Theorem \ref{thm:introgame}, we  deduce Topkis's theorem \cite[Theorem 4.2.1]{topkis1998supermodularity}, which requires topological assumptions.
\begin{cor}[Topkis]\label{cor:Topkis421}
	Assume the following conditions.
	\begin{enumerate}\item For every $i\in N$, there is an integer $m_i\ge0$ such that $S_i$ is a sublattice of $\R^{m_i}$.
		\item The subset $S$ is compact in $\R^{\sum_{i=1}^Nm_i}$.
		\item For every $i\in N$ and every $x_{-i}\in S_{-i}$, the function $f_i(\cdot, x_{-i}):S_i(x_{-i})\to \R$ is upper semicontinuous.\end{enumerate} Then, the set of Nash equilibria is a nonempty complete lattice.
\end{cor}
\begin{proof}
	From \cite[Theorem 2.3.1]{topkis1998supermodularity}, the sublattice $S\subset \prod_{i\in N}S_i$  is subcomplete. By Lemma \ref{lm:Euclidusc} (\ref{it:usc>orderusc}), for every  $x_{-i}\in S_{-i}$, the function $f_i(\cdot, x_{-i}):S_i(x_{-i})\to \R$ is order upper semicontinuous. By Theorem \ref{thm:introgame}, the set of Nash equilibria is a nonempty complete lattice.
\end{proof}
\section{Conclusion}
We use our generalization of Zhou's fixed point theorem to investigate the Nash equilibrium of supermodular games. We prove the existence of Nash equilibria and show that the set of Nash equilibria is a complete lattice. These result extend the celebrated Topkis  theorem. Compared with Topkis's theorem, our result does not need any topological condition, thereby enhancing the applicability and relevance of the results to various contexts.

	\section*{Statements and Declarations}

\subsection*{Conflict of interests}
The author declares that she has no conflict of interests.
\subsection*{Ethical approval}
This article does not contain any studies with human participants or animals performed by the author.
\subsection*{Informed consent} For this type of study informed consent was not required.

\subsection*{Acknowledgments}
I express gratitude to my  supervisor Philippe Bich for his constructive feedback and invaluable support. I am grateful to the referees for their comments and  suggestions.
\printbibliography
\end{document}